\documentclass[10pt]{article}
\usepackage{amstext,amsthm,amsfonts}
\usepackage{amsfonts}
\usepackage{amsmath}
\usepackage{geometry}
\usepackage[utf8]{inputenc}
\usepackage{enumerate}
\usepackage{array}
\usepackage{amssymb}
\usepackage{amsopn}
\usepackage{latexsym}
\usepackage{units}

\usepackage{mathrsfs}

\usepackage{xy}
\def\Si{\Sigma}
\def\si{\sigma}
\def\lb{\lambda}
\def\Op{\mathfrak{Op}}
\def\hb{\hbar}
\def\H{\mathcal{H}}
\def\de{\mathrm{d}}

\def\A{{\mathcal A}}
\def\B{{\mathrm B}}
\def\R{\mathbb R}
\def\C{\mathbb C}

\def\sp{\mathop{\mathrm{sp}}\nolimits}

\newtheorem{proposition}{Proposition}[section]
\newtheorem{Theorem}{Theorem}[section]
\newtheorem{Remark}{Remark}[section]
\newtheorem{Definition}{Definition}[section]
\newtheorem{Lemma}{Lemma}[section]
\newtheorem{Corollary}{Corollary}[section]

\providecommand{\keywords}[1]
{
  \small	
  \textbf{\textit{Keywords---}} #1
}

\begin{document}
\title{Constants of Motion of the Harmonic Oscillator.}

\date{\empty}
\author{Fabi\'an Belmonte \\ Universidad Católica del Norte \\\texttt{fbelmonte@ucn.cl}\and Sebasti\'an Cuéllar \\Universidad Católica del Norte \\ \texttt{sebastian.cuellar01@ucn.cl}}
\maketitle
\textbf{Address}: Departamento de Matem\'aticas, Universidad Cat\'olica del Norte, Angamos 0610, Antofagasta, Chile;
\vspace{1cm}
\begin{abstract}
We prove that Weyl quantization preserves constant of motion of the Harmonic Oscillator. We also prove that if $f$ is a classical constant of motion and $\Op(f)$ is the corresponding operator, then $\Op(f)$ maps the Schwartz class into itself and it defines an essentially selfadjoint operator on $L^2(\R^n)$. As a consequence, we provide detailed spectral information of $\Op(f)$. A complete characterization of the classical constants of motion of the Harmonic Oscillator is given and we also show that they form an algebra with the Moyal product. We give some interesting examples and we analize Weinstein average method within our framework.
\end{abstract}

\section{Introduction.}
Let $h_0$ be a classical Hamiltonian. A classical constant of motion for $h_0$ is a classical observable $f$ such that $\{h_0,f\}=0$, where $\{\cdot,\cdot\}$ is the Poisson bracket on the classical phase space. Similarly, if $H_0$ is a quantum Hamiltonian, a quantum constant of motion for $H_0$ is a quantum observable $F$ such that $[H_0,F]=0$, which means that the selfadjoint operators $H_0$ and $F$ strongly commute. Equivalently, a constant of motion is an observable that is invariant by the evolution of the system given by the Hamiltionian. Moreover, classical and quantum constants of motion admit decompositions through the reduction and diagonalization processes respectively (see Subsection \ref{com}). The analogy between classical and quantum constant of motion suggests that a canonical quantization $\Op$ should map classical constants of motion of $h_0$ into quantum constants of motion of $H_0=\Op(h_0)$. In such case we say $\Op$ preserves constants of motion of $h_0$. The initial motivation of this article is to address the latter problem when $h_0$ is the classical Harmonic Oscillator and $H_0$ is the quantum Harmonic Oscillator. \\ \\
Let $\mathfrak{Op}$ be the canonical Weyl quantization. The celebrated Groenewold-Van Hove's no go Theorem (see \cite{Gro}) implies that $\Op$ does not interchange the Poisson bracket of classical observables with the commutator of the corresponding operators. Therefore we should not expect that $\mathfrak{Op}$ preserves constants of motion for every $h_0$. However, it is not difficult to argue that the latter property might hold in some examples. As we mentioned before we will prove that the Harmonic Oscillator is one of those examples, i.e. \textit{$\Op$ preserves constants of motion of the Harmonic Oscillator}. We combine the latter result with the $N-$representation Theorem to obtain a large number of further consequences concerning both classical and quantum constants of motion. \\ \\
In Section \ref{Prel}, we sumarize some known facts concerning constants of motion and the classical and quantum Harmonic Oscillator. We also recall the average method of Weinstein for constructing constants of motion \cite{Wei}. \\ \\ In Section \ref{PresCons}, we prove that canonical Weyl quantization preserves constants of motion of the Harmonic Oscillator (Theorem \ref{PrCOM}). Let $f$ be a classical constant of motion of the Harmonic Oscillator. Using the $N-$representation Theorem, we prove that $\Op(f)$ maps the Schwartz class $S(\R^n)$ into itself (Theorem \ref{StoS}) Note that such result is usually obtained by imposing strong conditions on the symbol $f$, for instance belonging to a Hörmander class  (\cite{Horm}, \cite{Fol}, \cite{Ta}). Moreover we show that $\Op(f)$ is essentially selfadjoint on the Schwartz class $S(\R^n)$. Again, such result is usually obtained imposing strong conditions on $f$, for instance $\Op(f)$ is essentially selfadjoint if $f$ is hypoelliptic (Theorem 26.2 of \cite{Shu}). Since constants of motion are decomposed through the spectral diagonalization of the Hamiltonian, we  obtain detailed spectral information of $\Op(f)$ (Theorem \ref{selfcm}). We also provide an interesting criteria for boundedness of $\Op(f)$ (Proposition \ref{bounded}). Concerning the construction of constants of motion, we prove that Weyl quantization exchange Weinstein average of observables (Corollary \ref{wcom}). In addition, we show that the metaplectic representation maps the unitary group $U(n)$ into unitary quantum constants of motion, so we can also represent the Lie algebra $\mathfrak{u}(n)$ as quantum constants of motion Corollary \ref{metdecom} and \ref{metdecom1} (moreover, those operators are also the quantization of some classical constants of motion as well). \\ \\
In Section \ref{char}, we give a complete characterization of classical constants of motion. More precisely, we proof that every tempered constant of motion is of the form 
$$f=\displaystyle\sum_{|\alpha|=|\beta|}c_{\alpha,\beta} W(\phi_\beta, \phi_\alpha),$$
where $\alpha,\beta\in\mathbb{N}^{n}$, $\phi_\alpha$ and $\phi_\beta$ are the correspondig Hermite functions and $W$ is the Wigner transform. Depending on the behavior of the constants $c_{\alpha,\beta}$, the corresponding series will converges in different spaces ($S(\R^{2n})$, $L^2(\R^{2n})$ or $S'(\R^{2n})$) and moreover $c_{\alpha,\beta}=\langle\Op(f)\phi_\alpha,\phi_\beta\rangle$ (Theorem \ref{CCM_Th} and Proposition \ref{CCM_Pr}). We also use our techniques to give explicit expressions of the quantum and classical Weinstein averages (Propositions \ref{WeinsteinC} and \ref{WeinsteinQ}). In Subsection \ref{MSUB} we discuss some well known facts regarding the Moyal product. For example it is not always well defined when the two factors are tempered distributions. In order to avoid that problem the so called multiplier algebra  (see \cite{EGV}) was introduced. We prove that the set of tempered constants of motion form a subalgebra of the multiplier algebra and in particular the Moyal product between two tempered constants of motion is always well defined (Theorem \ref{MoyalMainT} and Corollary \ref{MoyalCor}). \\ \\
In Section 5, we give two explicit formulas to compute the Wigner transform $W(\phi_\alpha,\phi_\beta)$ and the coefficients $\langle\Op(f)\phi_\alpha,\phi_\beta\rangle$ (Theorem \ref{W_Form} and Theorem \ref{IntForm}). \\ \\
Finally, we include an appendix containing known facts concerning Weyl quantization that are needed for the comprehension of this article. 
\section{Preliminaries.}\label{Prel}
In this Section, we summarize some general facts concerning constants of motion and the classical and quantum Harmonic Oscillator. First let us fix some notation:
\begin{itemize}
\item $S(\R^m)$ is the Schwartz space of rapidly decreasing smooth functions on $\R^m$ with its canonical Fréchet space topology.
\item $S'(\R^m)$ denote the topological dual of $S(\R^m)$ known as the space of tempered distributions.
\item We endow $L^{2}(\R^m)$ with its cannonical inner product, linear on the left, and for $T\in S'(\R^m)$ and $\phi\in S(\R^m)$ we use the notation $\langle T,\phi\rangle=T(\overline{\phi})$.
\item If $H$ is a linear operator, we denote its domain by $D(H)$. 
\end{itemize}
\subsection{Constants of Motion.}\label{com}

In this Subsection we recall the concept of (classical and quantum) constant of motion and some of their properties. We shall emphasize the analogy between their classical and quantum description.

Let us fix a complete Hamiltonian $h_0\in C^\infty(\R^{2n})$ and denote by $\varphi_t$ its corresponding Hamiltonian flow.
We say that $f\in C^\infty(\R^{2n})$ is a classical constant of motion for $h_0$ if $\{h_0,f\} = 0$, where $\{\cdot,\cdot\}$ denotes the Poisson bracket corresponding to the canonical symplectic structure on $\R^{2n}$. Leibniz's rule and Jacobi identity show that the set $\A$ of all constants of motion is a Poisson subalgebra of $C^\infty(\R^{2n})$. It is easy to show that $f$ belongs to $\A$ if and only if $f\circ\varphi_t=f$, for each $t\in\R$. \\ \\
An important construction for the analysis of constants of motion is the following: Let $\lb$ be a regular value of $h_0$. Thus, the constant energy level set $\hat{\Si}_\lb:=h_0^{-1}(\lb)\subseteq \R^{2n}$ is a $(2n-1)$-submanifold invariant under $\varphi$. Moreover, considering $h_0$ as an equivariant moment map, the orbits space $\Si_\lb:=\hat{\Si}_\lb/\varphi$ can be endowed with the symplectic form given by Marsden-Weinstein-Meyer reduction \cite{AM,MW,Me,M} (our particular case is sometimes called Jacobi-Liouville theorem).

Since $f\in\A$ is constant on each orbit, we can associate to $f$ the field of smooth functions $f_\lb\in C^\infty(\Si_\lb)$, defined by

\begin{equation}\label{lb}
f_\lb([x,\xi])=f(x,\xi),
\end{equation}

where $(x,\xi)\in\hat\Si_\lb$ is any element in the orbit $[x,\xi]$.\\ \\
Clearly, if $a\in C^\infty(\R)$ and $f$ is a constant of motion, then $a\circ f$ is also a constant of motion. In particular, $a\circ h$ is a constant of motion for any $a\in C^\infty(\R)$. \\ \\
Another way to construct constants of motion is the following: Let $f\in C^\infty(\R^{2n})$, under certain assumptions, the function $\tilde f$ given by    
\begin{equation}\label{CA}
\tilde f(x,\xi)=\int_{\R} f\circ\varphi_t(x,\xi)\de t
\end{equation}
defines a classical constant of motion. We shall not provide general conditions to guarantee that $\tilde f$ is well defined for a generic $h_0$, instead we are going to give the details when $h_0$ is the Harmonic Oscillator in Subsection \ref{HO}.

Let us pass to the quantum description of constants of motion. Fix a selfadjoint operator $H_0$ on $L^2(\R^n)$ (a quantum Hamiltonian). A quantum observable, i.e. a selfadjoint operator $F$, is a quantum constant of motion for $H_0$, if $F$ strongly commutes with $H_0$. If $F$ is a bounded selfadjoint operator, then $F$ is a constant of motion iff $e^{itH_0} F e^{-itH_0}=F$, for every $t\in\R$. Moreover, the set of bounded constants of motion is the selfadjoint part of a von Neumann algebra (the commutant of $H_0$).   \\ \\
In analogy with the classical framework, the following construction is usefull to describe quantum constant of motion (see \cite{D} or \cite{BS} for details). Let $\si (H_0)$ the spectrum of $H_0$. There is a unique Borel measure $\eta$ (up to equivalence) on $\si (H_0)$, a unique measurable field of Hilbert spaces $\{\H(\lb)\}_{\lb\in\si (H_0)}$ (up to unitary equivalence on $\eta$-almost every fiber), and a unitary operator $T:L^2(\R^n)\to\int_{\sp(H_0)}^{\oplus}\H(\lambda)\de\eta(\lambda)$ such that
$$
[T a(H_0) u](\lambda)=a(\lb)(Tu)(\lambda)\,\,\,\forall u\in \text{Dom}(a(H_0)),
$$
where $a$ is any Borel function on $\si(H_0)$ and $a(H_0)$ denotes the corresponding operator given by the functional calculus. $T$ is called the diagonalization of $H_0$. In particular we have that
$$
[TH_0 u](\lambda)=\lambda(Tu)(\lambda)\,\,\,\forall u\in \text{Dom}(H_0).
$$
and
$$
[Te^{itH_0} u](\lambda)=e^{it\lb}(Tu)(\lambda)\,\,\,\forall u\in \H, t\in\R.
$$
Notice the similarities between the construction of $\H(\lb)$ and $\Si_\lb$: Both constructions are meant to make the Hamiltonian constant on each fiber and the corresponding dynamics trivial.\\ \\
It is well known that $F$ is a constant of motion if and only if it admits a decomposition through $T$ \cite{D,BS}, i.e. there is a measurable field of selfadjoint operators $\{ F_\lb\}_{\lb\in\si(H_0)}$ such that $[TFu](\lb)=F_\lb[Tu(\lb)]$. Such field of operators is the quantum counterpart of the field of classical observables $f_\lb$ defined in \eqref{lb}. \\ \\   
It is easy to prove that, if $F$ is a quantum constant of motion, so is $a(F)$ and $a(F)_\lb=a(F_\lb)$, for any Borel function $a$ on $\R$. \\ \\
As in the classical framework, we can construct constants of motion as follows: Given a selfadjoint operator $F$ on $L^2(\R^n)$, under certain conditions, we can define the quantum constant of motion $\tilde F$ given by
\begin{equation}\label{QA}
\tilde F=\int_{\R} e^{itH_0}Fe^{-itH_0}\de t.
\end{equation}
Just as in the classical case, we shall not give general conditions under which $\tilde F$ is a well defined operator, we will only explain in detail the case when $H_0$ is the quantum Harmonic Oscillator. 

The classical and quantum average trick above have been successfully used to study the asymptotic behaviour of clusters of eigenvalues for Schrodinger operators on compact Riemannian manifolds, for instance see \cite{Wei, G}. \\ \\
Using similar ideas, in \cite{DGW} it was shown that $\text{tr}(e^{-it(H_0+P)})$ has the same singularities than $\text{tr}(e^{-itH_0})$, where $P$ is an isotropic pseudodifferential operator of order $1$. We expect that our results concerning the constant of motion of the Harmonic Oscillator can be useful for further developments on the latter research topic.  
  

\subsection{The Harmonic Oscillator.}\label{HO}
This Subsection is meant to summarize some of the well known facts concerning the Harmonic Oscillator that we are going to need later. \\ \\ The classical Harmonic Oscillator is the physical system with hamiltonian given by $h_0(x,\xi)=\frac{1}{2}(\|x\|^2+\|\xi\|^2)$. It is easy to check that the flow of $h_0$ is given by
$$\varphi_t (x,\xi)=\begin{bmatrix}
(\cos t)I&(\sin t)I\\
(-\sin t)I&(\cos t)I
\end{bmatrix}\begin{bmatrix}
x\\
\xi
\end{bmatrix}$$
The flow $\varphi_t$ is a linear simplectomorphism. In particular, it preserves volume and its pullback maps $S(\R^n)$ into itself, so we can extend the pullback to $S'(\R^n)$. 
\begin{Definition}\label{TCOM}
For each $T\in S'(\mathbb{R}^{2n})$, we define the distribution $\varphi^*_t T$ by  
$$
(\varphi^*_t T)(f)=T(f\circ \varphi_t),
$$ 
where $f$ is any rapidly decreasing function on $\R^{2n}$. If $\varphi^*_t T=T$, we say that $T$ is tempered constant of motion of the Harmonic Oscillator.
\end{Definition}
With the identification of $(x,\xi)\in \mathbb{R}^{2n}$ with $x+i\xi\in\mathbb{C}^{n}$, the flow admits the representation 
$$
\varphi_t(x+i\xi)=e^{-it}(x+i\xi).
$$
If $\lambda$ is a regular value for $h_0$, that is $\lambda\neq 0$, then the submanifolds $h_0^{-1}(\lambda)$ are precisely the spheres $\mathbb{S}_{\sqrt{2\lambda}}^{2n-1}$. Hence the orbit space $\Sigma_{\lambda}=\mathbb{S}_{\sqrt{2\lambda}}^{2n-1} / \varphi$ is isomorphic as a manifold to the proyective plane $\C \mathbb{P}^{n-1}$. \\ \\
Since the flow is $2\pi$-periodic, we can replace it by the corresponding action of the group $\mathbb{T}$. In particular, we have no technical problems defining the average $\tilde{f}$, as in Subsection \ref{com}, by the Bochner integral
$$
\tilde{f}=\frac{1}{2\pi}\int_{\mathbb{T}}\varphi^*_t(f)\ dt,
$$
where $f$ is any tempered distribution and $\mathbb{T}$ is the circle group.  \\ \\
There is a family of polynomial constants of motion of particular interest in the present study. Set $z=x+i\xi=(x_1+i\xi_1,\dots,x_n+i\xi_n)=(z_1,\cdots,z_n)$ and $\alpha, \beta \in \mathbb{N}^{n}$ n-tuples of nonegative integers with $|\alpha|=|\beta|$. The monomial 
\begin{equation}\label{conmon}
m_{\alpha,\beta}(z)=z^{\alpha}\bar{z}^{\beta}
\end{equation}
is a constant of motion of the classical Harmonic Oscillator. For example $h_0$ itself and the coordinates of the angular momentum are a linear combination of some of the monomials $m_{\alpha,\beta}$. The set $$\{m_{\alpha,\beta}: |\alpha|=|\beta|=k\}$$ has exactly $d_k^2=\binom{n+k-1}{k}^2$ elements. We will show in Theorems \ref{CCM_Th} and \ref{W_Form} that the functions of this monomials generate the set of all classical constants of motion for the Harmonic Oscillator.\\ \\
The $n$ dimensional quantum Harmonic Oscillator is the self adjoint $H_0=\frac{1}{2}(-\Delta+\|x\|^{2})$ defined in a subspace of $L^2(\mathbb{R}^{n})$. The spectral decomposition of $H_0$ is well known: The spectrum  of $H_0$ is formed by degenerate eigenvalues of the form $k+\frac{n}{2}$, with $k$ any nonnegative integer, and the corresponding eigenspace $\mathcal{H}_k$ has dimension $d_k=\binom{n+k-1}{k}$. We shall denote by $P_k$ the orthogonal projection on $\mathcal{H}_k$. \\ \\ The $n$ dimensional  Hermite functions $\{\phi_{\alpha,n}:|\alpha|=k\}$ form a basis of $\mathcal{H}_k$, where
$$\phi_{\alpha,n}(x_1,...,x_n)=\phi_{\alpha_1,1}(x_{1})\cdots \phi_{\alpha_n,1}(x_{n})$$
and $\phi_{\alpha_i,1}$ is defined by 
$$\phi_{\alpha_i,1}(x)=(-1)^{\alpha_{i}}\frac{1}{\sqrt[4]{\pi(\alpha_{i}!2^{n})^{2}}}e^{\frac{x^{2}}{2}}\left(\frac{d^{\alpha_{i}}}{dx^{\alpha_i}}e^{-x^{2}}\right)$$
From now on we will not make explicit reference to the dimension of Hermite functions. The family $\{\phi_\alpha\}_{\alpha\in\mathbb{N}^{n}}$ is an orthonormal basis for $L^2(\mathbb{R}^{n})$, so 
$$L^2(\mathbb{R}^{n})=\bigoplus_{k\in\mathbb{N}}\mathcal{H}_{k}.$$
Clearly, the latter decomposition is the spectral diagonalization of $H_0$. It is well known that the one parameter group $e^{itH_{0}}$ is $2\pi$ periodic (we will show it in proposition \ref{teo 1}), so for bounded operators $F$ emulating equation \eqref{QA}, we can define the average $\tilde{F}$ by the Bochner integral
\begin{equation}\label{AV_HO}
\tilde{F}=\frac{1}{2\pi}\int_{\mathbb{T}}e^{itH_0}Fe^{-itH_0} \ dt.
\end{equation}
Latter we will give sense to this average for a wider class of operators.
\section{Preservation of constants of motion and its consequences.}\label{PresCons}
An equivalent version of the next result can be found in proposition 4.46 of \cite{Fol}. However it was not pointed out that $\varphi_t$ is the flow for the classical Harmonic Oscillator. 
\begin{proposition}\label{teo 1}
For each $t\in\R$, the metaplectic representation maps the flow $\varphi_t$ to the one parameter group $e^{itH_0}$. In particular, $e^{itH_0}$ maps $S(\R^n)$ into $S(\R^n)$ and it can be extended to a strongly continuous one parameter group on $S'(\R^n)$.
\end{proposition}
\begin{Corollary}\label{QA_Def}
If $F$ is a continuous linear operator from $S(\mathbb{R}^{n})$ to $S(\mathbb{R}^{n})$ then the average operator given by equation \eqref{AV_HO} is well defined in the strong sense. Similarly, if $F$ is a continuous linear operator from $S'(\mathbb{R}^{n})$ to $S'(\mathbb{R}^{n})$ then the average operator given by equation \eqref{AV_HO} is well defined in the weak sense. 
\end{Corollary}
Recall that $\Op$ denotes the Weyl quantization. In Appendix \ref{WQ} we summarize some of the main properties of Weyl quantization.
\begin{Theorem}\label{PrCOM}
Weyl quantization preserves constants of motion of the Harmonic Oscillator, more precisely, for each tempered constant of motion $f$ and each $t\in\R$, we have that 
\begin{equation}\label{PCM}
e^{itH_0}\Op(f)e^{-itH_0}=\Op(f).
\end{equation}
In particular, $[H_0,\Op(f)]=0$ on $S(\R^n)$.
\end{Theorem}
\begin{proof}
It is a direct consequence of Proposition \ref{teo 1} and equation \eqref{meta} in the Appendix \ref{WQ}.
\end{proof}
Using the Moyal product $\star$, we can find a weaker version of the latter theorem. For instance, since $h_0$ is quadratic, for suitable $f\in S'(\R^{2n})$ we have that $\Op(\{h_0,f\})=\Op(h_0\star f-f\star h_0)=[H_0,\Op(f)]$ as operators defined on $S(\R^n)$. In particular, if $f$ is a suitable tempered constant of motion $[H_0,\Op(f)]=0$. However, without using the metaplectic representation, it seems technically difficult to show that $e^{itH_0}$ maps $S(\R^n)$ into $S(\R^n)$, so \eqref{PCM} might not make sense, and even after showing that property, there are some extra technical difficulties to obtain \eqref{PCM} from the equation $[H_0,\Op(f)]=0$, which we do not know how to overcome directly. \\ \\
It is well known that $\Op$ defines a unitary isomorphism between $L^2(\R^{2n})$ and the space of Hilbert-Schmidt operators on $L^2(\R^n)$. In particular, for any $f\in L^2(\R^{2n})$, we have that $$\Op(f)=\sum_{\alpha,\beta}\langle\Op(f)\phi_\alpha,\phi_\beta\rangle\, P_{\alpha,\beta},$$ where the convergence holds in the space of Hilbert-Schmidt operators and $P_{\alpha,\beta}$ is the one dimensional rank operator given by $P_{\alpha,\beta}\phi=\langle\phi,\phi_\alpha\rangle\phi_\beta$. The particular properties of the Hermite orthonormal basis allow us to generalize such expression. 

\begin{Lemma}\label{Le1}
For any $f\in S'(\R^{2n})$, the kernel $K_f$ of $\Op(f)$ is given by 
\begin{equation}\label{kernelg}
K_f=\sum_{\alpha,\beta}\langle\Op(f)\phi_\alpha,\phi_\beta\rangle\phi_\beta\otimes \phi_\alpha,
\end{equation}
where $(\phi\otimes \psi)(x,y)=\phi(x)\psi(y)$. There is a constant $C>0$ and $\gamma\in \mathbb{N}^{2n}$ such that, for any $\alpha,\beta\in \mathbb{N}^n$, we have that
\begin{equation}\label{1ine}
|\langle\Op(f)\phi_\alpha,\phi_\beta\rangle|\leq C((\alpha,\beta)+1)^\gamma.
\end{equation}
Moreover, if $f\in S(\R^{2n})$, then for each $m\in\mathbb{N}$, we have that
\begin{equation}\label{2ine}
\sup_{\alpha,\beta\in \mathbb{N}^n}|\langle\Op(f)\phi_\alpha,\phi_\beta\rangle|\,(|\alpha|+|\beta|)^m< \infty.
\end{equation}
\end{Lemma}
\begin{proof}
Recall that $\Op(f):S(\R^n)\to S'(\R^n)$ is continuous. The $N$-representation theorem for $S(\R^n)$ (theorem V.13 \cite{RS}) implies that, for any $\phi\in S(\R^n)$, the convergence of the series $\sum \langle\phi,\phi_\alpha\rangle\phi_\alpha=\phi$ holds in $S(\R^n)$ endowed with its canonical locally convex topology. Therefore, for any $\phi,\psi\in S(\R^n)$ we have that
$$
\langle K_f,\psi\otimes\overline\phi\rangle= \langle\Op(f)\phi,\psi\rangle=\sum_{\alpha,\beta} \langle\Op(f)\phi_\alpha,\phi_\beta\rangle\langle \phi_\beta,\psi\rangle\langle \phi,\phi_\alpha\rangle=
\sum_{\alpha,\beta} \langle\Op(f)\phi_\alpha,\phi_\beta\rangle
\langle\phi_\beta\otimes\phi_\alpha, \psi\otimes\overline\phi\rangle.
$$
The inequalities \eqref{1ine} and \eqref{2ine} follows from the $N$-representation theorem for $S'(\R^{2n})$ and $S(\R^{2n})$ respectively (theorem V.13 and theorem V.14 in \cite{RS}).
\end{proof}
Inequalities \eqref{1ine} and \eqref{2ine} characterize elements in  $S'(\R^{2n})$ and $S(\R^{2n})$ respectively, but we shall leave the consequences of that fact for Section \ref{char}.\\ \\
We shall apply the previous results when $f$ is a tempered constant of motion of $h_0$. 
\begin{Corollary}\label{mainC}
Let $f$ be a tempered constant of motion.  If $|\alpha|\neq|\beta|$, then $\langle\Op(f)\phi_\beta,\phi_\alpha\rangle=0$. In particular, the kernel of $\Op(f)$ is given by
\begin{equation}\label{kernel}
K_f=\sum_{|\alpha|=|\beta|}\langle\Op(f)\phi_\alpha,\phi_\beta\rangle\phi_\beta\otimes \phi_\alpha.
\end{equation}
\end{Corollary}
\begin{proof}
Notice that
$$
\langle\Op(f)\phi_\alpha,\phi_\beta\rangle=\langle e^{itH}\Op(f)e^{-itH}\phi_\alpha,\phi_\beta\rangle=e^{it(|\alpha|-|\beta|)}\langle\Op(f)\phi_\alpha,\phi_\beta\rangle.
$$
Thus, $\langle\Op(f)\phi_\beta,\phi_\alpha\rangle=0$, unless $|\alpha|=|\beta|$. 
\end{proof}

The fact that the eigenspaces of $H_0$ are finite dimensional and the $N$-representation theorem implies the following remarkable result.
\begin{Theorem}\label{StoS}
If $f$ is a tempered constant of motion, then $\Op(f):S(\R^n)\to S(\R^n)$.
\end{Theorem}
\begin{proof}
By Corollary \ref{mainC} we know that for all $\varphi\in S(\R^n)$ we have
$$\Op{(f)}\varphi= \sum_{|\alpha|=|\beta|}\langle\Op(f)\phi_\alpha,\phi_\beta \rangle \langle \varphi, \phi_\beta\rangle  \phi_\alpha=\sum_{\alpha}\left(\sum_{|\beta|=|\alpha|}\langle\Op(f)\phi_\alpha,\phi_\beta \rangle \langle \varphi, \phi_\beta\rangle\right)  \phi_\alpha$$
Notice that there is $\gamma\in \mathbb{N}^{2n}$ and $C>0$, such that given $\alpha\in\mathbb{N}^n$ and for all $\beta\in\mathbb{N}^n$ with $|\alpha|=|\beta|$
$$|\langle\Op(f)\phi_\alpha,\phi_\beta \rangle \langle \varphi, \phi_\beta\rangle|\leq C\max_{|\beta|=|\alpha|}((\alpha,\beta)+1)^\gamma |\langle \varphi, \phi_\beta\rangle|\leq C (|\alpha|+1)^{|\gamma|} |\langle \varphi, \phi_\beta\rangle|.$$
Since $\varphi\in S(\R^n)$ and $d_{|\beta|}$ is a polynomial of degree $n-1$ on $|\beta|$, for all $m\in\mathbb{N}$ we have that
\begin{eqnarray*}
|\alpha|^{m}\left|\sum_{|\beta|=|\alpha|}\langle\Op(f)\phi_\alpha,\phi_\beta \rangle \langle \varphi, \phi_\beta\rangle\right|&\leq& \sum_{|\beta|=|\alpha|}|\alpha|^{m}(|\alpha|+1)^{|\gamma|} |\langle \varphi, \phi_\beta\rangle| \\
&\leq&\sum_{|\beta|=|\alpha|}\frac{1}{d_{|\beta|}} \sup_{\beta\in \mathbb{N}^{n}} \left(d_{|\beta|}|\beta|^{m}(|\beta|+1)^{|\gamma|} |\langle \varphi, \phi_\beta\rangle| \right)\\
&\leq& C
\end{eqnarray*}
By the $N-$representation Theorem $\Op(f)\varphi\in S(\R^n).$
\end{proof}
The following straightforward consequence is an example of how the results of this article can be applied in the Fock-Bargmann model.
\begin{Corollary}
Let $f$ be a tempered constant of motion and let $B$ be the Bargmann transform. The operator $B\Op(f)B^*$ maps homogeneous polynomials of degree $k$ into homogeneous polynomials of degree $k$ and has Wick symbol $f_W$ determined by
$$
f_W(z,\overline z)=2^n\int f(x,\xi)e^{-2\pi[(s-\xi)^2+(r-x)^2]}\de\xi\de x,
$$
where $z=r-is$.
\end{Corollary}
\begin{proof}
It is well known that $B\H_k$ is the space of homogeneous polynomials of degree $k$. The last statement follows from the previous theorem \ref{StoS} and proposition 2.97 in \cite{Fol}. 
\end{proof}
Recall that a tempered distribution $T$ is called real if for every $\varphi\in S(\R^{m})$ we have $T(\varphi)=\overline{T(\overline{\varphi}})$. The identity \eqref{WeylWigner} implies that if $f\in S'(\R^{2n})$ is real, then $$\langle\Op(f)\varphi,\psi\rangle=\langle\varphi,\Op(f)\psi\rangle.$$
However, even if $\Op(f)$ maps $S(\R^{n})$ into $S(\R^{n})$, the latter identity does not guarantee that $\Op(f)$ is a selfadjoint operator. The following Theorem implies that this is the case if, in addition, $f$ is a tempered constant of motion.
\begin{Theorem}\label{selfcm}
If $f$ is a real tempered constant of motion, then $\Op(f)$ is a essentially selfadjoint operator on $S(\R^{n})$. In addition the following identity holds
\begin{equation}\label{deco}
 \Op(f)=\bigoplus_k\Op^k(f)=\bigoplus_k\sum_{|\alpha|=|\beta|=k}\langle\Op_\hb(f)\phi_\alpha,\phi_\beta\rangle P_{\alpha,\beta},
\end{equation}
where  $\Op^k(f)$ is the restriction of $\Op(f)$ to $\H_k$. If  $\sigma(\Op(f))$ and $\sigma_{p}(\Op(f))$  denotes the spectrum and the point spectrum of $\Op(f)$ respectively, then 
$$
\sigma_{p}(\Op(f))=\bigcup_k\sigma_{p}(\Op^k(f)).
$$ 
In particular $\sigma_{p}(\Op(f))$ is at most numerable and it is dense in $\sigma(\Op(f))$. Moreover, $\Op(f)$ is bounded if and only if
$$
\sup\{|\lb|:\lb\in \sigma_{p}(\Op(f))\}<\infty.
$$
\end{Theorem}
\begin{proof}
Since $\Op^{k}(f)$ is a selfadjoint operator, $U_k(t):=e^{it\Op^k(f)}$ is a unitary one parameter group on $\mathcal{H}_k$. Define 
$$U(t):=\bigoplus_{k\in \mathbb{N}}U_k.$$ Clearly $U(t)$ is a unitary one parameter group, let us check it is also strongly continuous. Let $\varphi\in L^2(\R^{n})$ then 
$$\|U(t)\varphi-\varphi\|^{2}=\sum_{k=1}^{\infty}\|U_k(t)P_k\varphi-P_k\varphi\|^{2}.$$
Note that $\|U_k(t)P_k\varphi-P_k\varphi\|^{2}\leq 4\|P_k\varphi\|^{2}$, so by the Lebesgue dominated convergence theorem, $\displaystyle\lim_{t\to 0}\|U(t)\varphi-\varphi\|^{2}=0$, which proves $U$ is strongly continuous. By Stone's Theorem $U(t)$ has a selfadjoint infinitesimal generator $H$, we shall prove that $S(\R^{n})\subset D(H)$ and $H\varphi=\Op(f)\varphi$, for $\varphi\in S(R^n)$. Now observe that, by Fubini's Theorem
\begin{eqnarray*}
\int_{0}^{t}iU(a)\Op(f)\varphi \ da&=&\int_{0}^{t}\sum_{k=0}^{\infty}i e^{a\Op_k (f)}\Op_k (f)P_k\varphi \ da \\
&=&\sum_{k=0}^{\infty}\int_{0}^{t}i e^{a\Op_k (f)}\Op_k (f)P_k\varphi \ da\\
&=&\sum_{k=0}^{\infty}(U_k(t)-I)P_k\varphi \\
&=&(U(t)-I)\varphi.
\end{eqnarray*}
The last equality implies that the function $t\to U(t)\varphi$ is differentiable and $\displaystyle\frac{dU(t)\varphi}{dt}=iU(t)\Op(f)\varphi$ which proves our claim. The same argument in the proof of Theorem \ref{StoS} shows that $U(t)$ maps $S(\R^n)$ to $S(\R^n)$, so by Theorem VIII.10 of \cite{RS} $\Op(f)$ is essentially selfadjoint. \\ 
Identity \eqref{deco} is a direct consequence of Corollary \ref{mainC}.  Clearly $\displaystyle\bigcup_k\sigma_{p}(\Op^k(f))\subseteq \sigma_{p}(\Op(f)).$ Conversely, if $\Op(f)\varphi=\lambda\varphi$ with $\varphi\neq 0$, then $\Op^k(f)P_k\varphi=\lambda P_k\varphi$ for each $k\in\mathbb{N}$. Since $\varphi\neq 0$, there is $k_0$ such that $P_{k_0}\varphi\neq 0$, therefore $\lambda\in\sigma_{p}(\Op^{k_0}(f))$ and $\sigma_{p}(\Op(f))=\displaystyle\bigcup_{k}\sigma_{p}(\Op^k(f))$. It is well known that 
$$\sigma(\Op(f))=\overline{\bigcup_k\sigma(\Op^k(f))}=\overline{\bigcup_k\sigma_p(\Op^k(f))}.$$ 
Finally,  $\Op(f)$ is bounded iff $\displaystyle\sup_k{||\Op^k(f)||}=\sup_k\max\{|\lambda|:\lambda\in\sigma_p(\Op^k(f))\}<\infty$, and this completes the proof.
\end{proof}
\begin{Remark}
{\rm 
There exist tempered constants of motion $u_t$ such that $\Op(u_t)$ is the one parameter group $U(t)=e^{it\overline{\Op(f)}}$ (where $\overline{\Op(f)}$ is the closure of the operator $\Op(f)$). The latter claim follows from Theorem \ref{CCM_Th} in the next Section.
}
\end{Remark}
 The following result is a simple but interesting criteria for boundedness.
 \begin{proposition}\label{bounded}
If $f$ is a tempered constant of motion such that 
$$
\sup_{\alpha,\beta\in \mathbb{N}^n}|\langle\Op(f)\phi_\alpha,\phi_\beta\rangle| |\alpha|^{n-1}< \infty,
$$
then $\Op(f)$ defines a bounded operator on $L^2(\R^n)$.
\end{proposition}
\begin{proof}
Recall that $\Op(f)$ is bounded iff $\sup_k{||\Op^k(f)||}<\infty$. Since each $\H_k$ is finite dimensional, thus $\displaystyle\|\Op^k(f)\|\leq d_k\max_{|\alpha|=|\beta|=k}|\langle\Op(f)\phi_\alpha,\phi_\beta\rangle|$. Therefore, the fact that $d_k$ is a polynomial on $k$ of degree $n-1$ implies our result.
\end{proof}

The following important result is consequence of proposition \ref{teo 1} and equation \eqref{metaW}. 

\begin{Corollary}\label{wcom} Weyl quantization intertwine the averaging of observables, that is $\Op(\tilde f)=\widetilde{\Op(f)}$.
\end{Corollary}

Finally, using Proposition \ref{teo 1}, we shall provide some interesting examples of constants of motion coming from the restriction of the metaplectic representation $\mu$ to the complex unitary group $U(n)$. The following result is a direct consequence of the comments after the proof of Proposition 4.75 in \cite{Fol},  but we shall approach it in a different way (without considering the Bargmann transform). 
\begin{Corollary}\label{metdecom}
$\mu(U(n))$ is a unitary group of constant of motion of $H_0$.
\end{Corollary}
\begin{proof}
Recall that the flow $\varphi_t$ in complex coordinates is given by $\varphi_t=e^{it} I$. In particular, for each $S\in U(n)$, we have that
$$
e^{itH_0}\mu(S) e^{-itH_0}=\mu\left(\varphi_t\circ S\circ\varphi_{-t}\right)=\mu (S).
$$ 
\end{proof}
 The previous corollary allow us to define the representations $\mu_k:U(n)\to\mathcal U
(\H_k)$ given by $\mu_k (S)=\mu (S)|_{\H_k}$, for each $k\in\mathbb{N}$. Abusing of the notation, we also denote by $\mu$ the restriction  of the metaplectic representation to $U(n)$. Hence, $\mu$ admits the decomposition
$$
\mu=\bigoplus_{k}\mu_k.
$$
It is well known that each $\mu_k$ is irreducible (for instance \cite{Ig}). Moreover, since $d_k=\displaystyle\binom{n+k-1}{k}$, we expect that $\mu_1$ is equivalent to the canonical representation of $U(n)$ and $\mu_k$ is the corresponding symmetric $k$-th power representation.   

Let $\mathfrak{u}(n)$ the Lie algebra of $U(n)$ and $\pi$ be any strongly continuous unitary representation. For any $A\in \mathfrak{u}(n)$, we can define the strongly continuous one parameter group $U(t)=\pi(e^{itA})$. Stone's theorem implies $U(t)$ has a selfadjoint infinitesimal generator which we denote by $\de\pi(A)$. The map $\de\pi$ can be extended to the enveloping algebra of $\mathfrak{u}(n)$, but selfadjointness is not automatically guaranteed (ellipticity might be required, see chapter VI in \cite{Mau}). For simplicity, we shall not consider such general case in this article.

The following result is a direct consequence of corollary \ref{metdecom}. 
\begin{Corollary}\label{metdecom1}
For each $A\in \mathfrak{u}(n)$, the selfadjoint operator $\de\mu(A)$ is a constant of motion essentially selfadjoint on $S(\R^n)$. Moreover, the map $\de\mu$ admits the decomposition
$$
\de\mu =\bigoplus_{k}\de\mu_k.
$$
\end{Corollary}

The operators $\de\mu(A)$ also comes from Weyl quantization. Indeed, $\de\mu(A)=\Op(p_A)$, where $p_A$ is the polynomial of degree 2 given by
$$
p_A(w)=-\frac{1}{2}w\cdot A\mathcal{J}\cdot w,
$$
and $\mathcal{J}$ is the canonical symplectic matrix. Moreover, the flow of $p_A$ is $e^{tA}$ (see theorem 4.45 in \cite{Fol}). We would like to obtain refine spectral information concerning the operators $\de\mu(A)$ using the latter decomposition, corollary \ref{selfcm} and representation theory. However, we shall leave the latter problem open for a future work.  
\section{Characterization of classical constants of motion and their Moyal product.}\label{char}

Recall that $\Op$ defines an isomorphism between the space of tempered distributions $S'(\R^{2n})$ and the space of continuous operators $\B[S(\R^n),S'(\R^n)]$. Moreover,  if $T\in B[S(\R^n),S'(\R^n)]$ then $\Op^{-1}(T)=\tilde W(K_T)$, where $K_T$ is the kernel of $T$ given by the Schwartz's kernel theorem. 
\begin{proposition}\label{MP}
Let $\Phi^{\alpha,\beta}=W(\phi_\beta,\phi_\alpha)$. For each  
$t\in\R$, we have that 
\begin{equation}\label{flow}
\Phi^{\alpha,\beta}\circ\varphi_t=e^{it(|\alpha|-|\beta|)}\Phi^{\alpha,\beta}.
\end{equation}
 In particular, $\Phi^{\alpha,\beta}$ is a constant of motion iff $|\alpha|=|\beta|$. Moreover, $\Op(\Phi^{\alpha,\beta})=P_{\alpha,\beta}$. 
\end{proposition}
\begin{proof}
Proposition 4.28 in \cite{Fol} implies that  
$$
\Phi^{\alpha,\beta}\circ\varphi_t=W(\mu(\varphi_{-t})\phi_\beta,\mu(\varphi_{-t})\phi_\alpha)=e^{it(|\alpha|-|\beta|)}\Phi^{\alpha,\beta}.
$$
The identity $\Op(\Phi^{\alpha,\beta})=P_{\alpha,\beta}$ is a particular case of example (ii) in page 92 of \cite{Fol}.
\end{proof}
The following result is a direct consequence of Lemma \ref{Le1}.
\begin{Lemma}\label{Le2}
For any $f\in S'(\R^{2n})$, we have that
$$
f=\sum_{\alpha,\beta} \langle\Op(f)\phi_\alpha,\phi_\beta\rangle\Phi^{\alpha,\beta},
$$
where the convergence of the series holds in $S'(\R^{2n})$ with its canonical  topology. In addition, if $f\in S(\R^{2n})$ the convergence of the series holds in $S(\R^{2n})$.
\end{Lemma}
The following results provide a complete characterization of tempered, square integrable and rapidly decreasing constants of motion.
\begin{Theorem}\label{CCM_Th}
For any tempered constant of motion $f$, we have that
$$
f=\sum_{|\alpha|=|\beta|} \langle\Op(f)\phi_\alpha,\phi_\beta\rangle\Phi^{\alpha,\beta},
$$
where the convergence of the series holds in $S'(\R^{2n})$ with its canonical  topology. Conversely, if $C_{\alpha,\beta}$ is a family of constants such that  there is a constant $C>0$ and $\gamma\in \mathbb{N}^{2n}$ satisfying
$$
|C_{\alpha,\beta}|\leq C((\alpha,\beta)+1)^\gamma,
$$
then the series $\sum_{|\alpha|=|\beta|} C_{\alpha,\beta}\Phi^{\alpha,\beta}$ converges on $S'(\R^{2n})$ to a tempered constant of motion. Similarly,  if $C_{\alpha,\beta}$ is a family of constants such that, for each $m\in\mathbb{N}$, we have that 
$$
\sup_{\alpha,\beta\in I^n}{|C_{\alpha,\beta}|\,(|\alpha|+|\beta|)^m}< \infty,
$$
then the series $\sum_{|\alpha|=|\beta|} C_{\alpha,\beta}\Phi^{\alpha,\beta}$ converges on $S(\R^{2n})$ to a rapidly decreasing constant of motion.
\end{Theorem}.

The following result is the square integrable version of the latter theorem.
\begin{proposition}\label{CCM_Pr}
The set $\{\Phi^{\alpha,\beta}\}_{|\alpha|=|\beta|}$ forms a orthogonal basis of the Hilbert space of square integrable constants of motion.
\end{proposition}
\begin{proof}
We only need to check orthogonality, but this is a direct consequence of Proposition 1.92 in \cite{Fol}.
\end{proof}

The following result is another interesting consequence of equation \eqref{flow} and Lemma \ref{Le2}. 
\begin{proposition}\label{WeinsteinC}
For any $f\in S'(\R^{2n})$ we have that
$$
\hat f=\sum_{|\alpha|=|\beta|} \langle\Op(f)\phi_\alpha,\phi_\beta\rangle\Phi^{\alpha,\beta}.
$$
\end{proposition}
\begin{proof}
Notice that, 
$$
\int_{\mathbb{T}}\Phi^{\alpha,\beta}\circ\varphi_t \de t=\Phi^{\alpha,\beta}\int_{\mathbb{T}}e^{it(|\alpha|-|\beta|)}\de t.
$$ 
Thus, $\widetilde{\Phi^{\alpha,\beta}}=\Phi^{\alpha,\beta}$ when $|\alpha|=|\beta|$ and it vanishes otherwise. Therefore, the continuity of the average map implies that
$$
\tilde f=\sum_{\alpha,\beta} \langle\Op(f)\phi_\alpha,\phi_\beta\rangle\widetilde{\Phi^{\alpha,\beta}}=
\sum_{|\alpha|=|\beta|} \langle\Op(f)\phi_\alpha,\phi_\beta\rangle\Phi^{\alpha,\beta}.
$$
\end{proof}
The following result is the quantum version of the previous one.
\begin{proposition}\label{WeinsteinQ}
For any linear continuous operator $F:S'(\R^{2n})\to S'(\R^{2n})$ we have that
$$
\tilde{F}\phi_\alpha =P_{|\alpha|}F\phi_{\alpha}.
$$
Moreover, the kernel $K_{\tilde{F}}$ of $\tilde{F}$ is 
$$K_{\tilde{F}}=\sum_{|\alpha|=|\beta|}\langle F\phi_{\alpha},\phi_{\beta}\rangle \phi_{\beta}\otimes \phi_{\alpha}$$
\end{proposition}
\begin{proof}
Note that 
$$\int_{\mathbb{T}}e^{-it\left(H_{0}-k-\frac{n}{2}\right)}\phi_{\alpha} \ dt=\int_{\mathbb{T}}e^{-it\left(|\alpha|-k\right)}\phi_{\alpha} \ dt=\left[\int_{\mathbb{T}}e^{-it\left(|\alpha|-k\right)}dt\right]\phi_{\alpha}=2\pi\delta_{k,|\alpha|}\phi_{\alpha}.$$
Therefore 
$$\frac{1}{2\pi}\int_{\mathbb{T}}e^{-it\left(H_{0}-k-\frac{n}{2}\right)} \ dt=P_{k},$$
in consequence
$$\tilde{F}\phi_{\alpha}=\frac{1}{2\pi}\int_{\mathbb{T}}e^{itH_0}Fe^{-itH_0}\phi_\alpha \ dt=\frac{1}{2\pi}\int_{\mathbb{T}}e^{-it\left(H_{0}-|\alpha|-\frac{n}{2}\right)} F\phi_\alpha \ dt=P_{|\alpha|}F\phi_{\alpha}.$$
\end{proof}
\subsection{The Moyal product of constants of motion.}\label{MSUB}
In this Section we consider the so called Moyal product. It is clear that if $f\in S(\R^{2n})$, then $\Op(f):S(\R^n)\to S(\R^n)$. Then for each $f,g\in S(\R^{2n})$, there is $h\in S(\R^{2n})$ such that $\Op(f)\Op(g)=\Op(h)$. The Moyal product of $f$ and $g$ is by definition $f\star g=h$. It turns out that $\star$ defines a non-commutative associative product. For example, Proposition \ref{MP} implies that
\begin{equation}\label{wignermoyal}
\Phi^{\alpha,\beta}\star \Phi^{\alpha',\beta'}=\delta_{\alpha',\beta}\Phi^{\alpha,\beta'}
\end{equation}
Once Planck's constant $\hb$ is introduced (see Appendix \ref{WQ}), the Moyal product admits a power series expansion on $\hb$. In fact, the Moyal product was originally defined through that expansion (\cite{Gro},\cite{BFFLS}). The rigorous proof of equality between between the two expressions was established in \cite{EGV} (where the Moyal product is called twisted product). \\ \\ 
One of the main properties of the Moyal product is the identity
$$
\int_{\R^{2n}}f\star g(x,\xi)\de\xi\de x=\int_{\R^{2n}} fg(x,\xi)\de\xi\de x.
$$
The latter formula allow us to extend the definition of the Moyal product as follows: For each $T\in S'(\R^{2n})$ and $f\in S(\R^{2n})$, define $T\star f\in S'(\R^{2n})$ and $f\star T\in S'(\R^{2n})$ by
$$
T\star f(g)=T(f\star g),\qquad f\star T(g)=T(g\star f).
$$
However, we cannot define $T_1\star T_2$ for arbitrary $T_1,T_2\in S'(\R^{2n})$. In order to overcome the latter issue, it was introduced the so called Moyal multiplier algebra $\mathcal M$: we say that $T\in\mathcal M$ if $T\star f$ and $f\star T$ belong to $S(\R^n)$, for each $f\in S(\R^n)$. Therefore, for each $T_1,T_2\in\mathcal M$, we can define 
$$
T_1\star T_2(f)=T_1(T_2\star f).
$$
With the latter extension of the Moyal product and complex conjugation, $\mathcal M$ becomes a $\ast$-algebra. \\

The Corollary \ref{StoS} implies that it is always possible to define the compositions $\mathfrak{Op}(f)\mathfrak{Op}(g)$ and $\mathfrak{Op}(g)\mathfrak{Op}(f)$ for any pair of tempered constants of motion $f,g$. The latter suggest that it should be always possible to define the Moyal products $f\star g$ and $g\star f$. Indeed, we will prove that any tempered constant of motion is an element of the multiplier algebra.
\begin{Lemma}\label{LemaMoyal}
If $f\in S'(\mathbb{R}^{2n})$ and $g\in S(\R^{2n})$ are such that $f=\displaystyle\sum_{\alpha,\beta}c_{\alpha,\beta}\Phi^{\alpha,\beta}$ and $g=\displaystyle\sum_{\alpha,\beta}d_{\alpha,\beta}\Phi^{\alpha,\beta}$ then
\begin{equation}\label{moyalm}
\displaystyle f\star g=\sum_{\alpha,\beta}\left(\sum_{\alpha'}c_{\alpha,\alpha'}d_{\alpha',\beta}\right)\Phi^{\alpha,\beta}
\end{equation}
\end{Lemma}
\begin{proof}
This result is follows from formula \eqref{wignermoyal} and using twice the continuity of the Moyal product when one of the factors is fixed.
\end{proof}
\begin{Theorem}\label{MoyalMainT}
Every tempered constant of motion belongs to the multiplier algebra $\mathcal{M}$.
\end{Theorem}
\begin{proof}
Let $f$ be a tempered constant of motion and $g\in S(\mathbb{R}^{2n})$. There exists coefficients $a_{\alpha,\beta}$ and $b_{\alpha, \beta}$ such that
$$f=\sum_{|\alpha|=|\beta|}a_{\alpha,\beta}\Phi^{\alpha,\beta}\mbox{ and }g=\sum_{\alpha',\beta'}b_{\alpha',\beta'}\Phi^{\alpha',\beta'}.$$
Using equation \eqref{moyalm}
\begin{eqnarray*}
f\star g&=&\sum_{\alpha,\beta}\left(\sum_{|\alpha'|=|\alpha|}a_{\alpha,\alpha'}b_{\alpha',\beta}\right)\Phi^{\alpha,\beta}.\\
\end{eqnarray*}
Since $g\in S(\R^{2n})$ and $d_{|\alpha|}$ is a polynomial on $|\alpha|$ of degree $n-1$, equations \eqref{1ine} and \eqref{2ine} imply that for all $m\in\mathbb{N}$
\begin{eqnarray*}
\sum_{|\alpha'|=|\alpha|}|a_{\alpha,\alpha'}b_{\alpha',\beta}|(|\alpha|+|\beta|)^{m}&\leq& \sum_{|\alpha'|=|\alpha|}\frac{|a_{\alpha,\alpha'}|}{d_{|\alpha|}(|\alpha|+1)^{|\gamma|}} \sup_{\alpha',\beta}\left(d_{|\alpha'|}(|\alpha'|+1)^{|\gamma|}(|\alpha'|+|\beta|)^{m}|b_{\alpha',\beta}|\right)\\
&\leq& C \sum_{|\alpha'|=|\alpha|}\frac{((\alpha,\alpha')+1)^{\gamma}}{d_{|\alpha|}(|\alpha|+1)^{|\gamma|}} \\
&\leq& C
\end{eqnarray*}
Therefore Lemma \ref{Le1} implies $f\star g\in S(\R^{2n})$. A similar argument proves that $g\star f\in S(\R^{2n})$.
\end{proof}
\begin{Corollary}\label{MoyalCor}
The space of tempered constants of motion of the Harmonic Oscillator form a $\ast$-algebra with the Moyal product. 
\end{Corollary}
\section{Explicit computations for $\Phi^{\alpha,\beta}$ and $\langle\Op(f)\phi_{\alpha},\phi_{\beta}\rangle$.}
Expressions for $\Phi^{\alpha,\beta}$ can be found in the literature for the case $n=1$ for example in \cite{Fol}. The following result is a closed formula for Wigner transform of $n$ dimensional Hermite functions. 
\begin{Theorem}\label{W_Form}
If $|\alpha|=|\beta|$ we have the following formula
$$
\Phi^{\alpha,\beta}(x,\xi)=2^{n}\pi^{\frac{n}{2}}e^{-h_0(x,\xi)} m_{\alpha',\beta'}\left(\xi+2i\pi x\right)\prod_{k=1}^{n}P_{\alpha_k, \beta_k}\left(\frac{1}{\sqrt{2}}\left\|\xi_k+2i\pi x_k\right\|\right)
$$
with $m_{\alpha',\beta'}$ some monomial of the form \eqref{conmon} and $P_{\alpha_i,\beta_i}$ some polynomial. In consequence, the classical constants of motion of $h_0$ are generated by the monomials $m_{\alpha,\beta}$ and the functions of $\|x_k+i\xi_k\|$.
\end{Theorem}
\begin{proof}
We are interested in a closed expression for $W(\phi_\alpha,\phi_\beta)$ with $|\alpha|=|\beta|$. Due to the integral definition of the Wigner transform and its relation to the Fourier Wigner Transform (Section 1.4. in \cite{Fol}) we get
$$W(\phi_\alpha,\phi_\beta)(x,\xi)=\prod_{k=1}^{n}W(\phi_{\alpha_i,1},\phi_{\beta_i,1})(x_i,\xi_i)=2^n(-1)^{|\alpha|}\prod_{k=1}^{n}V(\phi_{\alpha_k,1},\phi_{\beta_k,1})(x_k,\xi_k).$$ 
There are several expressions for $V(\phi_{\alpha_k,1},\phi_{\beta_i,1})(x_k,\xi_k)$, one is given in \cite{FW} in terms of complex Hermite polynomials. After some easy computations we find that
$$W(\phi_\alpha,\phi_\beta)(x,\xi)=2^{n+|\alpha|}\pi^{\frac{n}{2}}e^{-\frac{\|z\|^2}{2}} \prod_{k=1}^{n}H_{\alpha_k, \beta_k}(z_k)$$
The complex Hermite polynomials admit a factorization of the form $H_{\alpha_k,\beta_k}(z_k)=O_{\alpha_k,\beta_k}^{|\alpha_k-\beta_k|}(z_k)P_{\alpha_k, \beta_k}(\|z_k\|)$ where $P_{\alpha_k, \beta_k}$ is the polynomial
$$P_{\alpha_k, \beta_k}(\|z_k\|)=\sum_{s=0}^{\min(\alpha_k,\beta_k)}(-1)^s s! \binom{m}{s} \binom{n}{s} \|z_k\|^{\min(\alpha_k,\beta_k)-s}$$
and 
$$O_{\alpha_k,\beta_k}(z_k)=\begin{cases}
z_k&\alpha_k>\beta_k\\
\overline{z_k}&\alpha_k\leq\beta_k
\end{cases}$$ 
The latter implies for some indices $\alpha',\beta'\in\mathbb{N}^{n}$ with $|\alpha'|=|\beta'|$ we have
$$W(\phi_\alpha,\phi_\beta)(x,\xi)=2^{n+|\alpha|}\pi^{\frac{n}{2}}e^{-\frac{\|z\|^2}{2}} m_{\alpha',\beta'}(z)\prod_{i=1}^{n}P_{\alpha_i, \beta_i}(\|z_i\|)$$
\end{proof}

We shall look for an integral expression of $\langle\Op(f)\phi_\alpha,\phi_\beta\rangle$. With this purpose in mind, in what follows we assume that $f$ is a measurable polynomially bounded function so it defines a tempered distribution.  \\ \\
In order to follow the general perspective described in Section 2.1, we shall denote by $\C\mathbb{P}_\lb^{n-1}$ the projective plane obtained as the orbit space of the flow $\varphi_t$ restricted to the sphere $\mathbb{S}^{2n-1}_{\sqrt{2\lb}}$.  Recall that, since each $\varphi_t$ is an isometry, $\C\mathbb{P}^{n-1}_\lb$ admits a canonical Riemannian structure called the Fubini-Study metric, which by definition makes the projection $\pi:\mathbb{S}^{2n-1}_{\sqrt{2\lb}}\to \C\mathbb{P}^{n-1}_\lb$ a Riemannian submersion. The following identity is a direct consequence of proposition A.III.5 in \cite{BGM}.
\begin{Lemma}
Let $\nu^\lb_S$ the canonical volume form of the sphere $\mathbb{S}^{2n-1}_{\sqrt{2\lb}}$ and $\nu^\lb_F$ the Fubini-Study volume form of the projective plane $\C\mathbb{P}^{n-1}_\lb$. Also let $\nu^\lb_{[z]}$ the volume form of the orbit $[z]$ (a great circle on $\mathbb{S}^{2n-1}_{\sqrt{2\lb}}$) induced from the its canonical Riemannian structure. For any $g\in L^1(\mathbb{S}^{2n-1}_{\sqrt{2\lb}})$ we have that
$$
\int_{\mathbb{S}^{2n-1}_{\sqrt{2\lb}}}g(z)\nu^\lb_S(z)=\int_{\C\mathbb{P}^{n-1}_\lb}\left(\int_{[z]}g|_{[z]}\nu^\lb_{[z]}\right)\nu^\lb_F([z])
$$
\end{Lemma}
\begin{Theorem}\label{IntForm}
Let $f\in\A$. If $|\alpha|=|\beta|$, then 
$$
<\Op(f)\phi_\alpha,\phi_\beta>=2\pi\int_0^\infty \left(\int_{\mathbb{C}\mathbb{P}^{n-1}_\lb}f_\lb\Phi^{\beta,\alpha}_\lb\nu^\lb_F\right)\de\lb.
$$
\end{Theorem}

\begin{proof}
Using coarea formula \cite{Fe, Fe0,Si} we have that
$$
\langle\Op(f)\phi_\alpha,\phi_\beta \rangle=\int_{\R^{2n}}f(x,\xi) \Phi_{\beta,\alpha}(x,\xi)\de x\de\xi= \int_0^\infty\left(\int_{\mathbb{S}^{2n-1}_{\sqrt{2\lb}}} f(z)\Phi^{\beta,\alpha}( z)\nu^\lb_S(z)\right)\frac{1}{\sqrt{2\lb}}\de\lb.
$$
Since $f\Phi^{\beta,\alpha}$ is constant on each orbit, the previous lemma implies that 
$$
\int_{\mathbb{S}^{2n-1}_{\sqrt{2\lb}}} f(z)\Phi^{\beta,\alpha}( z)\nu^\lb_S(z)=2\pi\sqrt{2\lb}\int_{\C\mathbb{P}^{n-1}_\lb} f_\lb([z])\Phi_\lb^{\beta,\alpha}([z])\nu^\lb_F([z])
$$
and this finish the proof.
\end{proof}

\begin{Remark}
{\rm
Since all the projective planes are diffeomorphic, we can replace the integration over $\mathbb{C}\mathbb{P}^{n-1}_\lb$ by integration over a single projective plane, but we would need to include a suitable Jacobian. Indeed, we have that
\begin{equation}\label{ie}
<\Op(f)\phi_\alpha,\phi_\beta>=2^n\pi\int_0^\infty \left(\int_{\mathbb{C}\mathbb{P}^{n-1}}f(\lb z)\Phi^{\beta,\alpha}(\lb z)\nu_F([z])\right)\lb^{n-1}\de\lb.
\end{equation}
}  
\end{Remark}
\begin{Remark}
{\rm
Let $B$ be the Bargmann transform. Then $B(\H_k)$ is the space of homogeneous polynomials of degree $k$. Thus, $B(\H_k)$ can be identified with the Hilbert space $\Gamma^0(\C\mathbb{P}^{n-1},\mathcal O(k))$ of global holomorphic Sections of the vector bundle $\mathcal{O}(k)\to\C\mathbb{P}^{n-1}$. Since the operator $B\Op(f)B^*$ maps $\Gamma^0(\C\mathbb{P}^{n-1},\mathcal O(k))$ into itself, $\langle\Op(f)\phi_\alpha,\phi_\beta\rangle=\langle B\Op(f)B^*(B\phi_\alpha), B\phi_\beta\rangle$ can be expressed as an integral over $\C\mathbb{P}^{n-1}$, somehow exchanging the double integral \eqref{ie} in the previous remark. 
}
\end{Remark}
\appendix
\section{Weyl Quantization.}\label{WQ}
In this Appendix we recall the definition of Weyl quantization and some of its main features. \\ \\
Weyl quantization \cite{We, Gro} (or Weyl calculus) is a map meant to transform real functions in the canonical phase space $\R^{2n}$ (classical observables) into selfadjoint operators on $L^2(\R^n)$ (quantum observables) in a physically meaningful manner, that is, taking into account the analogies between the descriptions of classical and quantum mechanics. There are several approaches to introduce Weyl quantization, for instance \cite{Z,Ta, Shu}, but we will mainly follow \cite{Fol}. We also recommend \cite{TE1} for a quite complete review of quantization theory. \\ \\
Formally, for certain function $f$ on phase space, we define the operator $\Op_\hb(f)$ acting on $L^2(\R^n)$ given by
\begin{equation}\label{We}
\Op_\hb(f)u(x)=\int_{\R^n}\int_{\R^n}f\left(\frac{x+y}{2},\xi\right)e^{\frac{i}{\hb}(x-y)\cdot \xi}u(y)\de\xi\de y, 
\end{equation}
where $\hb$ is a positive parameter interpreted as Planck's constant. In this article, we will not need to consider the roll played by $\hb$, so we will take $\hb=1$ and $\Op:=\Op_1$.\\ \\
The meaning of the expression \eqref{We} changes depending on the class of functions we are considering. More precisely, if $f\in S'(\R^{2n})$, then $\Op(f):S(\R^n)\to S'(\R^n)$, so \eqref{We} defines a tempered distribution (in particular, the evaluation on $x$ is not well defined). The standard way to rigorously define $\Op(f)$ is the following: Notice that, the kernel of $\Op(f)$ should be
\begin{equation}\label{Ker}
K_f(x,y)=(I\otimes \mathcal F) f\left(\frac{x+y}{2},y-x\right)=\int_{\R^n}f\left(\frac{x+y}{2},\xi\right)e^{i(x-y)\cdot \xi}\de\xi,
\end{equation}
where $(I\otimes \mathcal F)$ is the Fourier transform in the momentum variable. The change of variables $(x,y)\to \left(\frac{x+y}{2},y-x\right)$ is linear, injective and preserves measure, thus composing with this map send $S(\R^{2n})$ to itself continuously and it can be extended  to $S'(\R^{2n})$. Since the $(I\otimes \mathcal F)$ is well defined on $S'(\R^{2n})$, we can define $K_f$ for $f\in S'(\R^{2n})$ using the middle expression in \eqref{Ker}. In fact, the map $f\to K_f$ is an automorphism of the locally convex space $S'(\R^{2n})$. Finally, for $u,v\in S(\R^n)$, we define
$$
\Op(f)u(v)= K_f(v\otimes u) 
$$  

A fundamental object in the description of Weyl quantization is the the so called Wigner transform $W:S(\R^n)\times S(\R^n)\to S(\R^{2n})$ defined by 

$$
W(u,v)=\int_{\R^n}e^{-i\xi\cdot p}u\left(x+\frac{p}{2}\right)\overline{v\left(x-\frac{p}{2}\right)}\de p
$$  
 
The Wigner transform can be regarded as the restriction to functions of two variables of the form $g(x,y)=u(x)\overline{v(y)}$ of the map $\tilde W$ defined by
$$
\tilde W g(x,\xi)=\int_{\R^n}e^{-i\xi \cdot p}g\left(x+\frac{p}{2}, x-\frac{p}{2}\right).
$$ 
Using the arguments in the definition of the kernel $K_f$, we can define $\tilde W$ on $S'(\R^{2n})$; in fact, $\tilde W$ is the inverse of the map $f\to K_f$ and $\tilde W$ is unitary on $L^2(\R^{2n})$. In particular, we have that
\begin{equation}\label{WeylWigner}
\langle\Op(f)u,v\rangle=\langle f,W(v,u)\rangle. 
\end{equation}
For instance, if $f$ is polynomially bounded, then
$$
\langle Op(f)u,v\rangle =\int_{\R^{2n}}f\, W(u,v).
$$
One of the main properties of Weyl quantization is its relation with the so called metaplectic representation. Let $Sp(n)$ be the real symplectic group, i.e. the group formed by all the linear and symplectic maps $S:\R^{2n}\to\R^{2n}$. There is a map $\mu: Sp(n)\to\mathcal U(L^2(\R^n))$, characterized up to a $\pm 1$ factor by the identity:
$$
\rho(S(q,p))=\mu(S)\rho(q,p)\mu(S)^{-1},
$$
where $(q,p)\in\mathbb{H}_n$ is an element of Heisenberg group and $\rho$ is the Schrödinger representation. The map $\mu$ is called the metaplectic representation (see \cite{Fol} for details). As consequences, we have the very important (and beautiful) identities
\begin{equation}\label{metaW}
W(\mu(S)\phi,\mu(S)\psi)=W(\phi,\psi)\circ S^*
\end{equation}
\begin{equation}\label{meta}
\Op(f\circ S^*)=\mu(S)\Op(f)\mu(S)^{-1}
\end{equation}
The metaplectic representation is not a true representation of $Sp(n)$, because in general we only have that $\mu(ST)=\pm\mu(S)\mu(T)$, so it defines a representation of the double covering group $Sp_2(n)$ of $Sp(n)$. However, restricted to the complex unitary group $U(n)$, $\mu$ does define a representation (notice that the Hamiltonian flow of the Harmonic Oscillator belongs to $U(n)$).


\begin{thebibliography}{00}




\bibitem{AM} R. Abraham and E. Marsden, \emph{Foundations of Mechanics}. Second edition. Benjamin/Cummings Publishing Co., Reading, Mass., (1978).

\bibitem{FW} F. Agorram, A. Benkhadra, A. El Hamyani. A. Ghanmi, \emph{Complex Hermite functions as Fourier Wigner transform} Integral Transforms and Special Functions  Volume 27. Issue 2 (2016)

\bibitem{BFFLS} F. Bayen, M. Flato, C. Fronsdal, A. Lichnerowicz, and D. Sternheimer, \emph{Deformation theory and quantization. I. Deformations of symplectic structures.} Ann. Phys. {\bf 111} (1978), no. 1, 61-110. \emph{II. Physical applications.} Ann. Phys. {\bf 111} (1978), no. 1, 111-151.
\bibitem{BGM} M. Berger, P. Gauduchon and E. Mazet, \emph{Le spectre d'une vari\'et\'e Riemanniene}. Lecture Notes in Math. Vol. 194, Springer, Berlin, (1971).

\bibitem{BS} M.S. Birman and M.Z. Solomyak, Spectral Theory of Selfadjoint Operators in Hilbert Spaces. Translated from the 1980 Rusian Original by S. Khrushchëv and V. Peller. Mathematics and its Applications (Soviet Series), Reidel Publishing Co., Dordrecht (1987).


\bibitem{D} J. Dixmier, \emph{Von Neumann Algebras}. Translated by F. Jellett from \emph{Les alg\`{e}bres d'operators dans
l'espace hilbertien(alg\`{e}bres de Von Neumann)}. North-Holland mathematical library, {\bf 27} (1969).

\bibitem{DGW}
M. Doll, O. Gannot and J. Wunsch, \emph{Refined Weyl law for homogeneous perturbations of the Harmonic Oscillator}. Comm. Math. Phys. {\bf 362} (2018), no. 1, 269-294.


\bibitem{EGV} R. Estrada, J.M. Gracia-Bond\'ia and J.C V\'arilly \emph{On Asymptotic Expansions of Twisted Products}. J. Math. Phys. {\bf 30} (1989), no. 12, 2789-2796.


\bibitem{Fe0} H. Federer, \emph{Curvature measures}. Trans. Amer. Math. Soc. {\bf 93}, (1959), 418–491.

\bibitem{Fe} H. Federer, \emph{Geometric Measure Theory}. Die Grundlehren der mathematischen Wissenschaften Band {\bf153},
Springer-Verlag, Berlin, (1969).

\bibitem{Fol} G. B. Folland, \emph{Harmonic Analysis in Phase Space}. Annals of Mathematics Studies, {\bf 122}. Princeton University Press, Princeton, NJ, (1989).


\bibitem{Gro} H.J. Groenewold, \emph{On the Principles of Elementary Quantum Mechanics}. Physica {\bf 12} (1946), 405-460.

\bibitem{G} V. Guillemin, \emph{Band asymptotics in two dimensions}. Adv. in Math. {\bf 42} (1981), no. 3, 248-282.


\bibitem{Horm} L. Hörmander,  \emph{Pseudo-differential operators and hypoelliptic equations}. Singular integrals (Proc. Sympos. Pure Math., Vol. X, Chicago, Ill., 1966), pp. 138–183. Amer. Math. Soc., Providence, R.I., (1967). 



\bibitem{Ig} J. Igusa, \emph{Theta Functions}. Springer-Verlag, New York, (1972).



















\bibitem{Mau} K. Maurin, \emph{General Eigenfunction expansions and unitary representations of topological groups}. Monografie Matematyczne, Tom 48 PWN-Polish Scientific Publishers, Warsaw, (1968).


\bibitem{M} J.E. Marsden, \emph{Lectures on Mechanics}. LMS Lecture Notes {\bf 174}. Cambridge: Cambridge University Press, (1992).

\bibitem{MW} J.E. Marsden and A. Weinstein, \emph{Reduction of Symplectic Manifolds with Symmetry}. Rep. Math. Phys.
{\bf 5} (1974), 121-130.



\bibitem{Me} K.R. Meyer, \emph{Symmetries and Integrals in Mechanics}. Dynamical Systems, M. Peixoto, Academic Press, (1973), 259-273.









\bibitem{Shu} M. A. Shubin, \emph{Pseudodifferential Operators and Spectral Theory}, Translated by S. I. Andersson, Springer Series in Soviet Mathematics, Berlin, Springer, (1987)

\bibitem{Si} L. Simon, \emph{Lectures on Geometric Measure Theory}. Proceedings of the Centre for Mathematical Analysis,
Australian National University, {\bf 3} (1983).

\bibitem{RS} M. Reed and B. Simon, \emph{Methods of Modern Mathematical Physics, volume I}. Academic Press, inc., San Diego, (1980). 



\bibitem{Ta} M. E. Taylor, \emph{Pseudodifferential Operators}. Princeton Mathematical Series, {\bf 34}, Princeton University Press, Princeton, N.J., (1981).


\bibitem{TE1} S. Twareque Ali and M. Engli\v{s}, \emph{Quantization Methods: a Guide for Physicists and Analysts}. Rev. Math. Phys. {\bf 17} (2005), no. 4, 391-490.




\bibitem{Wei} A. Weinstein, \emph{Asymptotics of eigenvalue clusters for the Laplacian plus a potential}. Duke Math. J. {\bf 44} (1977), no. 4, 883-892.


 \bibitem{We} H. Weyl, \emph{Quantenmechanik und Gruppentheorie}. Z. Physik {\bf 46} (1927), 1-46.




\bibitem{Z} M. Zworski, \emph{Semiclassical Analysis}. Graduate Studies in Mathematics {\bf 138}, AMS, Providence, RI, (2012).




















\end{thebibliography}
\end{document}